\documentclass[final,3p,times]{elsarticle}
\usepackage{algorithm}
\usepackage{algorithmic}

\usepackage{amsmath}
\usepackage{amsfonts}
\usepackage{amssymb}
\usepackage{amsthm}
\usepackage{hyperref}
\hypersetup{colorlinks=false}
\newtheorem{theorem}{Theorem}[section]
\newtheorem{lem}[theorem]{Lemma}

\newtheorem{cor}[theorem]{Corollary}
\theoremstyle{definition}
\newtheorem{defn}[theorem]{Definition}

\newtheorem{exmp}[theorem]{Example}

\theoremstyle{remark}
\newtheorem{rem}[theorem]{Remark}
\begin{document}
\begin{frontmatter}

\title{Detecting Simultaneous Integer Relations for Several Real Vectors}

 \author[label1,label2,label4]{Jingwei Chen}\ead{velen.chan@163.com}

\author[label1,label2]{Yong Feng \corref{cor}
}\ead{yongfeng@casit.ac.cn}
 \cortext[cor]{Corresponding author. }
\author[label1,label2,label4]{Xiaolin Qin}
\author[label1,label2]{Jingzhong Zhang}

\address[label1]{
UESTC-CICACAS Joint Laboratory of Automated Reasoning, University of
Electronic Science and Technology of China, Chengdu 610054}

\address[label2]{Lab. for Automated Reasoning and
      Programming, Chengdu Institute of Computer Applications, CAS, Chengdu 610041, China}

\address[label4]{Graduate University, Chinese Academy of Sciences, Beijing 100049,
China}

\begin{abstract}
An algorithm which either finds an nonzero integer vector ${\mathbf
m}$ for given $t$ real $n$-dimensional vectors ${\mathbf
x}_1,\cdots,{\mathbf x}_t$ such that ${\mathbf x}_i^T{\mathbf m}=0$
or proves that no such integer vector with norm less than a given
bound exists is presented in this paper. The cost of the algorithm
is at most ${\mathcal O}(n^4 + n^3 \log \lambda(X))$ exact
arithmetic operations in dimension $n$ and the least Euclidean norm
$\lambda(X)$ of such integer vectors. It matches the best complexity
upper bound known for this problem. Experimental data show that the
algorithm is better than an already existing algorithm in the
literature. In application, the algorithm is used to get a complete
method for finding the minimal polynomial of an unknown complex
algebraic number from its approximation, which runs even faster than
the corresponding \emph{Maple} built-in function.
\end{abstract}

\begin{keyword}
integer relation \sep PSLQ \sep HJLS \sep algebraic number \sep
minimal polynomial
\end{keyword}

\end{frontmatter}

\section{Introduction}
\label{sec:introduction}

Given a real vector ${\mathbf x}=(x_1,\cdots,x_n)^T\in{\mathbb
R}^n$, say a nonzero vector ${\mathbf
m}=(m_1,\cdots,m_n)^T\in{\mathbb Z}^n$ is an \emph{integer relation}
for $\mathbf{x}$ if ${\mathbf x}^T{\mathbf m}=\sum_{i=1}^nx_im_i=0$.
How to detect an integer relation for a given real vector is an old
problem. This is solved, for instance, by the PSLQ algorithm
\cite{FBA1999} that together with related lattice reduction schemes
such as LLL \cite{LLL1982}, was named one of ten ``algorithms of the
twentieth century'' by the publication \emph{Computing in Science
and Engineering} (see \cite{DS2000}). This paper considers a
generalization of the problem.  Let
 ${\mathbf x}_1, \cdots, {\mathbf x}_t$ be $t$ vectors in ${\mathbb R}^n$,
and denote $({\mathbf x}_1, \cdots, {\mathbf x}_t)$ by $X$. A
\emph{simultaneous integer relation} (SIR)  for ${\mathbf x}_1,
\cdots, {\mathbf x}_t$ is a vector ${\mathbf m} \in {\mathbb
Z}^{n}\setminus\{{\mathbf 0}\}$ such that $X^T{\mathbf m} = {\mathbf
0}$, i.e. ${\mathbf x}_i^T\mathbf{m} = 0$ for $i = 1,\cdots, t$. For
short, we also call ${\mathbf m}$  an SIR for $X$. An algorithm
which either finds an SIR  for $t$ real $n$-dimensional vectors or
proves that no SIR with norm less than a given bound
 exists is presented in this paper.

When $t=1$, the problem of detecting integer relations for one
rational or real vector is quite old. For two numbers $(a_1, a_2)$,
the venerable Euclidean algorithm does the job by computing the
ordinary continued fraction expansion of the real number $a_1/a_2$.
For $n \geq 3$, many detecting algorithms under the names
generalized Euclidean algorithm and multidimensional continued
fraction algorithm were proposed. We refer the reader to
\cite{HJL1989, FBA1999} for comprehensive surveys. Among these
integer relation algorithms, the LLL-based HJLS algorithm
\cite{HJL1989} and the PSLQ algorithm \cite{FBA1999} have been used
frequently.

To authors' known, the first algorithm to detect SIRs for several
real vectors ($t\geq 2$) was presented in \cite{HJL1989}, in which
J. Hastad, B. Just, J. C. Lagarias, and C. P. Schnorr not only
presented the HJLS algorithm to find integer relations for one real
vector, gave the first rigorous proof of a `polynomial time' bound
for a relation finding algorithm, but also proposed a simultaneous
relations algorithm. Unfortunately HJLS has a serious drawback: it
is extremely  unstable numerically (see \cite{FB1992, FBA1999}).

 In their draft \cite{RS1997}, C. R{\"o}ssner and C. P.
Schnorr proposed an algorithm which computes for real vectors
${\mathbf x}_1,\,{\mathbf x}_2$ simultaneous diophantine
approximation to the plane spanned by the vectors ${\mathbf
x}_1,\,{\mathbf x}_2$ by using a modified HJLS algorithm. It can be
seen as a special case  $t=2$ of the aforementioned problem. But for
the moment, it is still in a preliminary state with some open
problems.

The PSLQ algorithm \cite{FBA1999} is now extensively used in
Experimental Mathematics, with applications such as identification
of multiple zeta constants, a new formula for $\pi$, 
 quantum field theory
 and so on (see \cite{BB2001, BBC2007, BB2009}).
PSLQ employs a numerically stable matrix reduction procedure, so it
is numerically stable in contrast to other integer relation
algorithms. Moreover, it can be  generalized to the complex number
field and the Hamiltonian quaternion number field, but the
corresponding outputs are in Gaussian integer ring and Hamilton
integer ring respectively. For example, PSLQ will output $(1,
\,I,\,-I)^T$ for the complex vector $(1+I, \,1+2I,\,2+I)^T$, where
$I=\sqrt{-1}$. The reason is that Hermite reduction in PSLQ produces
some Gaussian integers in the reducing matrix (see Section
\ref{subsec:Hyperplane-matrix}). Thus PSLQ can not be used to detect
SIRs (in ${\mathbb Z}^n$) for several real vectors.

An algorithm to detect SIRs for $t$ real vectors   is presented in
this paper. It uses a technique,  similar to that in HJLS, to
construct the hyperplane matrix  and a method, generalized from
PSLQ, for matrix reduction. The algorithm either finds an SIR for
$X$ if one exists or proves that there are no SIRs for $X$ of norm
less than a given size.  The cost of the algorithm is at most
${\mathcal O}(n^4 + n^3 \log \lambda(X))$ exact arithmetic
operations to detect an SIR for $X$, where $n$ is the dimension of
the input real vectors and $\lambda(X)$ represents the least
Euclidean norm of SIRs for $X$.
Although the same theoretic complexity as obtained for the HJLS
simultaneous relations algorithm is proved , experiments show that
the algorithm in this paper often perfoms better in practice.
Furthermore, in contrast to PSLQ, our algorithm can be applied to
detect an integer relation in ${\mathbb Z}^n$ (rather than in the
Gaussian or Hamiltonian integer rings) for complex or Hamiltonian
vectors. Consequentially, a complete method to find the minimal
polynomial of an approximately complex algebraic number is obtained
by applying our algorithm.

Our main contributions in this paper are the following:
\begin{itemize}
\item We present a new algorithm to detect SIRs for several real
vectors and show that its complexity matches the best one known for
this problem (HJLS simultaneous relation algorithm).
\item We implement our algorithm in \emph{Maple} by two schemes. The one
uses software floating arithmetic (multiprecision floating point
arithmetic) in all steps, and the other partially uses software
floating arithmetic and mainly uses hardware floating arithmetic.
Then we report many experimental results, which shows that our
algorithm is relevant.
\item We successfully apply  our algorithm to find
the minimal polynomial of an approximately complex algebraic number.
This strategy is different from some known LLL-based methods, such
as \cite{KLL1988,HN2010}, and is for all complex algebraic numbers
rather than mere for real agebraic numbers in \cite{QFC2009}. We
also present many experiments, which shows that this newly complete
method is efficient and even better than the \emph{Maple}
built-in function \texttt{PolynomialTools:-MinimalPolynomial}.

\end{itemize}

 The reminder of this paper is organized as follows. In Section
 \ref{sec:Main algorithm},  both  preliminaries and main results of this paper
 are presented. The cost
 of our algorithm is  analyzed in Section \ref{sec:Correctness and Cost}. Some
 empirical studies, further discussions and an application of our algorithm  are included in Section \ref{sec:Empirical study}.

\section{The Main Algorithm}\label{sec:Main algorithm}

\subsection{Notations and Assumptions}

Throughout this paper, ${\mathbb Z}$, ${\mathbb R}$, and ${\mathbb
C}$ stand for the sets of integers,  real numbers, and complex
numbers respectively. For $c\in{\mathbb R}$, $\lfloor c \rceil$
denotes an arbitrary integer closest to $c$, i.e. $\lfloor c \rceil
= \lfloor c+\frac{1}{2}\rfloor$. All vectors in this paper are
column vectors, and will be denoted in bold.  If ${\mathbf x} \in
{\mathbb R}^n$, then ~$\|{\mathbf x}\|_2$~represents its Euclidean
norm, i.e. $\|{\mathbf x}\|_2 = \sqrt{\langle{\mathbf x},\, {\mathbf
x}\rangle}$, where~$\langle*,\,
*\rangle$ is the inner product of two vectors. We denote the $n\times n$
identity matrix by $I_{n}$. Given a matrix $A = (a_{i,j})$, we
denote its transpose by $A^T$, its trace by tr$(A)$, its determinant
by $|A|$, and its Frobenius norm by
 $\|A\|_F = (\mbox{tr}(AA^T))^{1/2} = (\sum
a_{i,j}^2)^{1/2}$. We say that a matrix $A$ is lower trapezoidal if
$a_{i,j} = 0$ for $i < j$. The group of $n\times n$ unimodular
matrices with entries in ${\mathbb Z}$ are denoted by $GL(n,
{\mathbb Z})$.

In what follows we always suppose that ${\mathbf x}_1, \cdots,
{\mathbf x}_t\in{\mathbb R}^n$ are linearly independent, where
${\mathbf x}_i = (x_{i,1}, \cdots, x_{i,n})^T$ and   $t < n$.
\footnote{Our assumption $t<n$ is based on the following fact: Any
SIR ${\mathbf m}$    is in the orthogonal complement space of
span$({\mathbf x}_1,\, \cdots,\,{\mathbf x}_t)$. Since
   ${\mathbf x}_1,\, \cdots,\,{\mathbf x}_t$ are linearly
   independent vectors in $\mathbb{R}^n$ we have $t\leq n$.
   So if $t= n$, then the dimension of the linear space span$({\mathbf x}_1,\, \cdots,\,{\mathbf x}_t)$ is $n$, hence that
   there exists no simultaneous integer relations for ${\mathbf x}_1,\, \cdots,\,{\mathbf x}_t$. }
   Obviously, every ${\mathbf x}_i$ is nonzero. Let $X\in
{\mathbb R}^{n\times t}$ be the matrix $({\mathbf x}_1, \cdots,
{\mathbf x}_t)$ and suppose that $X$ satisfies
\begin{equation}\label{eq:suppose-on-x} \left|
  \begin{array}{cccc}
    x_{1,n-t+1} & x_{2,n-t+1} & \cdots & x_{t,n-t+1} \\
    x_{1,n-t+2} & x_{2,n-t+2}  & \cdots & x_{t,n-t+2}\\
    \vdots &\vdots& &\vdots\\
    x_{1,n} & x_{2,n}  & \cdots & x_{t,n}
  \end{array}
\right| \neq 0.\end{equation} If $X \in {\mathbb R}^{n\times t}$
does not satisfy (\ref{eq:suppose-on-x}) we  can always (since $X$
has rank $t$) exchange some rows of $X$ to produce $X' = CX$ such
that $X'$
satisfying (\ref{eq:suppose-on-x})
, where $C$ is an appropriate matrix in $GL(n, {\mathbb Z})$. In
this case, we detect an SIR for $X'$. If $\mathbf{m}$ is detected as
an SIR for $X'$, then $C^T{\mathbf m}$ is an SIR for $X$.

\subsection{A Method to construct a Hyperplane Matrix}\label{subsec:Hyperplane-matrix}

\begin{defn}[Hyperplane Matrix]\label{def:Hyperplane-matrix}
Let~$X=({\mathbf x}_1,\cdots, {\mathbf x}_t) \in {\mathbb
R}^{n\times t}$. A hyperplane matrix  with respect to $X$ is any
matrix $H \in {\mathbb R}^{n\times(n-t)}$
 such that $X^TH = \mathbf{0}$ and the columns of $H$ $\mbox{span}$
$X^{\perp} = \{ {\mathbf y}\in{\mathbb R}^n: {\mathbf x}_i^T{\mathbf
y} = 0, i = 1, \cdots, t\}$.
\end{defn}

Given $X = ({\mathbf x}_1, \cdots, {\mathbf x}_t) \in {\mathbb
R}^{n\times t}$ satisfying (\ref{eq:suppose-on-x}). We now present a
method to construct a hyperplane matrix for  $X$. The basic idea is
from HJLS \cite{HJL1989}. The same strategy was also used in PSLQ,
based on a partial-sum-of-squares vector   and a lower-quadrature
matrix factorization, instead of Gram-Schmidt orthogonalization.

Let ${\mathbf b}_1, \cdots, {\mathbf b}_n$ form the standard basis
of ${\mathbb R}^{n}$, i.e. the $i$-th entry of ${\mathbf b}_i$ is
$1$ and others are $0$. Perform the process of standard Gram-Schmidt
orthogonalization to ${\mathbf x}_1, \cdots, {\mathbf x}_t,{\mathbf
b}_1,\cdots,{\mathbf b}_n$  in turn producing ${\mathbf x}_1^*,
\cdots, {\mathbf x}_t^*,{\mathbf b}_1^*,\cdots,{\mathbf b}_n^*$.
 Note  that, since $X$ satisfies (\ref{eq:suppose-on-x}), we have ${\mathbf b}_{n-t+1}^* = \cdots = {\mathbf b}_{n}^* =
\textbf{0}$.

Define $H_{X}$ to be the $n\times (n-t)$ matrix $({\mathbf b}_1^*,
\cdots, {\mathbf b}_{n-t}^*)$. From the following lemma,
$H_X=({\mathbf b}_1^*, \cdots, {\mathbf b}_{n-t}^*)$ is a hyperplane
matrix with respect to $X \in {\mathbb R}^{n\times t}$.

\begin{lem}\label{lem:property-of-Hx}
Let $X\in{\mathbb R}^{n\times t}$ and $H_{X}$ be as above. Then
\begin{enumerate}
\item
 $H_{X}^T H_{X} = I_{n-t}$.
\item
$\|H_{X}\|_F = \sqrt{n-t}$.
\item
$\left({\mathbf x}_1^*, \cdots, {\mathbf x}_t^*, H_{X}\right)$ is an
orthogonal matrix.

\item $X^T H_{X} = $\textbf{0}, i.e. $H_{X}$ is a hyperplane matrix of
$X$.
\item
$H_{X}$ is a lower trapezoidal matrix and every diagonal element of
$H_{X}$ is nonzero.
\end{enumerate}
\end{lem}

\begin{proof}
Since every two columns of $H_{X}$ are orthogonal, part $1$ follows.
And  part $2$ follows from part $1$. Let $X^* = ({\mathbf x}_1^*,
\cdots, {\mathbf x}_t^*)^T$. Obviously, $({\mathbf x}_1^*, \cdots,
{\mathbf x}_t^*,H_{X})$ is an orthogonal matrix. From part 3 and
standard Gram-Schmidt orthogonalization we have $X^{*T}H_{X} =
\textbf{0}$ and $X = X^*Q$ respectively, where $Q$ is an appropriate
$t\times t$ invertible matrix. Thus $X^TH_{X} = Q^TX^{*T}H_{X}=
\textbf{0}$ and hence that part 4 follows. We now prove part 5.
Denote the $k$-th element of ${\mathbf b}_i^*$ by $b_{i,k}^*$. The
diagonal elements of $H_{X}$ are $b_{i,i}^{*}$ for $i = 1, \cdots,
n-t$. Before normalizing ${\mathbf b}_i^*$ we have $ b_{i,i}^* = 1 -
\sum_{k=1}^{t}x_{k,i}^{*2} - \sum_{j=1}^{i-1}b_{j,i}^{*2}$, and at
the same time, \[0\neq\|{\mathbf b}_{i}^*\|_2^2=\langle{\mathbf
b}_{i}^*, {\mathbf
b}_{i}^*\rangle=1-\sum_{k=1}^{t}x_{k,i}^{*2}-\sum_{j=1}^{i-1}b_{j,i}^{*2}.\]
Thus all the diagonal elements of $H_{X}$ are nonzero. Now we only
need to show that $H_{X}$ is lower trapezoidal. From standard
Gram-Schmidt orthogonalization, we can check that $b_{i,k}^* =
\langle{\mathbf b}_i^*, {\mathbf b}_k^*\rangle = 0$ holds for $i>
k$. This completes the proof.
\end{proof}

\subsection{Generalized Hermite Reduction}\label{subsec:Matrix-Reudce}
We now study how to reduce the hyperplane matrix $H_X$. First we
recall (modified) Hermite reduction as presented in \cite{FBA1999}.

\begin{defn}[Modified Hermite reduction] \label{def:Modified-Hermite-Reduction}
Let $H=(h_{i,j})$ be a lower trapezoidal matrix with $h_{j,j} \neq
0$ and set $D := I_n$. For $i$ from $2$ to $n$, and for $j$ from
$i-1$ to $1$ by step $-1$, set $q := \lfloor h_{i,j}/h_{j,j}\rceil$;
then  for $k$ from $1$ to $n$, replace $d_{i,k}$ by $d_{i,k} -
qd_{j,k}$. We say $DH$ is the modified Hermite reduction of $H$ and
$D$ is the reducing matrix of $H$.
\end{defn}

If the entries of $H$ are complex numbers, then $q = \lfloor
h_{i,j}/h_{j,j}\rceil$ may be a Gaussian integer. Thus for a complex
vector, PSLQ can only gives  a Gaussian integer relation.

 Hermite reduction is also presented in
\cite{FBA1999}, and is equivalent to modified Hermite reduction  for
a lower triangular matrix $H$ with $h_{j,j} \neq 0$.  Both of the
two equivalent reductions have the following properties:
\begin{enumerate}
\item
The reducing matrix $D \in GL(n, {\mathbb Z})$.
\item  For all $k > i$, the (modified) Hermite reduced matrix $H' =
(h'_{i,j}) = DH$ satisfies $|h'_{k,i}| \leq |h'_{i,i}|/ 2 =
|h_{i,i}|/ 2$.
\end{enumerate}

Unfortunately, (modified) Hermite reduction is not suitable to
detect SIRs any more because it does not deal with the last $t-2$
rows of $H_{X}$ when $2<t<n$. In order that the reduced and reducing
matrices of $H_{X} \in{\mathbb R}^{n\times(n-t)}$ satisfy the two
properties above, we
 generalize the Hermite reduction as follows.

\begin{defn}[Generalized Hermite
Reduction]\label{def:Generalized-Hermite-reduction} Let $H$ be a
lower trapezoidal matrix with $h_{j,j} \neq 0$ and set $D := I_n$,
$H'=(h'_{i,j}):=H$. For $i$ from $2$ to $n$,
 for $j$ from $\min\{i-1, \, n-t\}$ by $-1$ to $1$, $q :=
\lfloor h'_{i,j}/h'_{j,j}\rceil$; for $k$ from $1$ to $j$,
$h'_{i,k}=h'_{i, k}-qh'_{j,k}$; for $k$ from $1$ to $n$, $d_{i,k} :=
d_{i,k} - qd_{j,k}$. For every two integers $s_1, s_2 \in \{n-t+1,
\cdots, n\}$ satisfying $s_1 < s_2$, $h'_{s_1, n-t} = 0$ and
$h'_{s_2, n-t} \neq 0$, exchange the $s_1$-th row and the $s_2$-th
row of $D$. We call $DH$ the generalized Hermite reduction of $H$
and $D$ the reducing matrix.
\end{defn}

Obviously, generalized Hermite reduction is equivalent to modified
Hermite reduction when $t=1$. In addition, we can easily check that
generalized Hermite reduction retains the two properties mentioned
above when $1\leq t<n$.

There are two main differences between the (modified) Hermite
reduction and the generalized Hermite reduction. Firstly, the last
$t-1$ rows of $H$ will also be reduced by the first $n-t$ rows of
$H$ in the generalized Hermite reduction, while the (modified)
Hermite reduction not. Secondly, generalized Hermite reduction
exchanges the $s_1$-th row and the $s_2$-th row of $D$ if $s_1 <
s_2$, $h_{s_1, n-t} = 0$ and $h_{s_2, n-t} \neq 0$ hold. This
implies that if $h_{n-t+1, n-t}= 0$ after  generalized Hermite
reduction then $h_{n-t+2, n-t} = \cdots = h_{n, n-t} = 0$.

\subsection{The Algorithm Description}
\label{subsec:The algorithm description}

Based on the method to construct the hyperplane matrix and the
generalized Hermite reduction,  an algorithm to detect SIR for real
vectors is proposed as follows.\\ \\
 \textbf{Algorithm 1 (Simultaneous Integer Relation Detection).}\\
 {\bf Input:} $({\mathbf
x}_1, \cdots, {\mathbf x}_t) = X\in{\mathbb R}^{n\times t}$
satisfying (\ref{eq:suppose-on-x}) and a parameter $\gamma >
2/\sqrt{3}$.\\
\textbf{Initiation.} \begin{itemize} \item Compute \emph{the
hyperplane matrix} $H_X$ and set $H := H_{X}$, $ B := I_n$.
\item
 Reduce the hyperplane matrix $H$ by \emph{the generalized Hermite
reduction} producing the reducing matrix $D$. Set $X^T :=
X^TD^{-1},\,H := DH,\,B:=BD^{-1}$.
\end{itemize}
\textbf{Iteration.} \begin{enumerate} \item Exchange. Let $H =
(h_{i,j})$. Choose an integer $r$ such that $\gamma^{r}|h_{r,r}|
\geq \gamma^{i}|h_{i,i}|$ for $1\leq i\leq n - t$.
 Let
\begin{equation}\label{eq:alph-beta-lambda-delta}
\begin{array}{ll}
  \alpha := h_{r,r}, & \beta := h_{r+1,r}, \\
  \lambda := h_{r+1, r+1}, & \delta := \sqrt{\beta^2 +
\lambda^2}.
\end{array}\end{equation}
Define the permutation matrix $R$ to be the identity matrix with the
$r$ and $r+1$~rows exchanged. Update $X^T := X^TR, H := RH,  B :=
BR$.
\item Corner.
 Let $Q := I_{n-t}$. If
$r < n - t$,
 then let the submatrix of $Q$ consisting of the $r$-th and
$(r + 1)$-th rows of columns $r$ and $r + 1$ be
$\left[\begin{array}{cc} \beta/\delta & -\lambda/\delta\\
\lambda/\delta&\beta/\delta
\end{array}
\right]$. Update $H := HQ$.

\item Reduction. Reduce $H$ by \emph{the generalized Hermite
reduction} producing $D$. Set $X^T := X^TD^{-1},\,H :=
DH,\,B:=BD^{-1}$.

\item Termination. Compute $G := 1/\|H\|_F$.
Then there  exists no SIR whose Euclidean norm is less than $G$.
Denote $B=({\mathbf B}_1,\cdots,{\mathbf B}_n)$, where ${\mathbf
B}_j\in{\mathbb R}^n$. If $X^T{\mathbf B}_j ={\mathbf 0}$ for some
$1 \leq j \leq n$, or $h_{n-t,n-t}=0$ then
\end{enumerate}
{\bf Output}: the corresponding SIR for $X$.

\begin{rem}
The description of Algorithm 1 is similar to PSLQ. But the biggest
difference is that Algorithm 1 used the generalized Hermite
reduction. The (modified) Hermite reduction used in PSLQ is not
suitable to detect SIRs for several real vectors, as mentioned
early. PSLQ may be viewed as a particular case of Algorithm 1 when
$t = 1$.
\end{rem}

\begin{rem}
Given a complex vector ${\mathbf z} = {\mathbf x} + {\mathbf y}I$ in
${\mathbb C}^n$ where ${\mathbf x}, {\mathbf y} \in {\mathbb R}^n$
and $I = \sqrt{-1}$, finding an integer relation (in ${\mathbb
Z}^n$) for ${\mathbf z}$ is equivalent to finding an SIR for
$({\mathbf x}, {\mathbf y})$. Thus Algorithm 1 can be used. For
instance, let ${\mathbf z} =(2+3I,4+9I,8+27I,16+81I,32+243I)^T$. For
finding an integer relation for ${\mathbf z}$, first let ${\mathbf
x}=(2,4,8,16,32)^T$ and ${\mathbf y}=(3,9,27,81,243)^T$. Running
Algorithm 1 with $\gamma = 1.16$, and ${\mathbf x}, {\mathbf y}$ as
its input vectors gives an SIR $(6, 7, -9, 2, 0)^T$ for $({\mathbf
x}, {\mathbf y})$, which, of course,  also is an integer relation
for ${\mathbf z}$. This is one of the biggest differences between
Algorithm 1 and PSLQ since  PSLQ  can only give a Gaussian integer
relation in ${\mathbb Z}[I]^n$ for ${\mathbf z}$ rather than an
integer relation in ${\mathbb Z}^n$.
\end{rem}

\begin{rem}
Generally,
 $t$ real $n$-dimensional vectors may have 0, 1, or up to $n-t$
linearly independent SIRs. One can follow the strategy in
\cite[Section 6]{FBA1999} to find them.
\end{rem}

\begin{theorem}\label{thm:Correctness}
Let $X=({\mathbf x}_1,\cdots,{\mathbf x}_t)$ satisfy
(\ref{eq:suppose-on-x}) and $\lambda(X)$ be the  least Euclidean
norm of SIRs for $X$. Suppose there exists an SIR for $X$. Then
\begin{enumerate}
\item An SIR for $X$ will appear as a column of $B$ after no more than

\[
\left[\left(\begin{array}{c}
              n \\
              2
            \end{array}
\right)-\left(\begin{array}{c}
              t \\
              2
            \end{array}
\right)\right] \frac{\log(\gamma^{n-t}\lambda(X))}{\frac{1}{2}
\log\left(\frac{4\gamma^2}{\gamma^2+4}\right)}.
\]
iterations in Algorithm 1.
\item If after a number of iterations no SIR has yet appeared in a column of
$B$, then there are no SIRs of norm less than the bound $1/\|H\|_F$.
\end{enumerate}
\end{theorem}

From this theorem, Algorithm 1 either finds an SIR ${\mathbf m}$ for
given real vectors ${\mathbf x}_1,\cdots,{\mathbf x}_t$ such that
${\mathbf x}_i^T{\mathbf m}=0$ or proves that no small simultaneous
integer relation with Euclidean norm less than $1/\|H\|_F$ exists.

Moreover, it can be proved that the norm of the SIR for $X$ output
by Algorithm 1 is no greater than $\gamma^{n-t-1}\lambda(X)$. This
is an important property of Algorithm 1, and the proof is similar to
that of Theorem 3 in \cite{FBA1999}.

\begin{cor}\label{cor:time-complexity-of-simultaneous-relation}
If $X\in{\mathbb R}^{n\times t}$ has SIRs, then there exists a
$\gamma$ such that Algorithm 1 can find an SIR for $X$ in polynomial
time ${\mathcal O}(n^4+n^3\log\lambda(X))$.
\end{cor}

\begin{proof}
Let $\gamma = 2$. Then Algorithm 1 constructs an SIR for $X$  in no
more than
\[
(n-t)^2(n+t-1)+(n-t)(n+t-1)\log\lambda(X)
\]
iterations.  Algorithm 1 takes ${\mathcal O}(n-t)$ exact arithmetic
operations per iteration, and hence that ${\mathcal
O}((n-t)^4+(n-t)^3\log\lambda(X))$ exact arithmetic operations are
enough to produce an SIR for $X$. Since $t<n$, the proof is
complete.
\end{proof}

\begin{rem}
All conclusions above also hold for complex numbers with
$\gamma>\sqrt{2}$, but the outputs of the corresponding variation of
Algorithm 1 are in Gaussian integer ring.
\end{rem}

\section{Proof of Theorem \ref{thm:Correctness}}\label{sec:Correctness and Cost}

Given $X = ({\mathbf x}_1, \cdots, {\mathbf x}_t)\in{\mathbb
R}^{n\times t}$, let $H_X$ be the hyperplane matrix obtained by the
method introduced in section \ref{subsec:Hyperplane-matrix} and let
$P_X=H_XH_X^T$.  By expanding this expression, it follows that
$P_{X} = I_n - \sum_{i=1}^t{\mathbf x}_i^* {\mathbf x}_i^{*T}$. Let
${\mathbf m}\in{\mathbb Z}^n$ be an SIR for $X$. Then it can be seen
that $P_X{\mathbf m}={\mathbf m}$ and $\|P_{X}\|_F = \sqrt{n-t}$.
For any matrix $D\in GL(n,{\mathbb Z})$ and $(n-t)\times (n-t)$
orthogonal matrix $Q$,
\begin{equation}\label{eq:DH_XQ...}\begin{array}{ll}
1\leq\|D{\mathbf m}\|_2=\|DP_X{\mathbf
m}\|_2&\leq\|DP_X\|_F\|{\mathbf m}\|_2\\&=\|DH_X\|_F\|{\mathbf
m}\|_2=\|DH_XQ\|_F\|{\mathbf m}\|_2,
\end{array}
\end{equation}
where $\|DP_X\|_F=\|DH_X\|_F$ follows from $P_X^T=P_X$ and
$P_X^2=P_X$. From (\ref{eq:DH_XQ...}), the part 2 of Theorem
\ref{thm:Correctness} follows.

Let $H(k)$ be the result after  $k$ iterations of Algorithm 1.

\begin{lem}\label{lem:Hn,n-2=0}
If $h_{j,j}(k) = 0$ for some $1 \leq j \leq n - t$ and no smaller
$k$, then $j = n - t$ and an SIR for $X$ must appear as a column of
the matrix $B$.
\end{lem}
\begin{proof}
By the hypothesis on $k$, all diagonal elements of $H(k - 1)$ are
not zero. Now, suppose the $r$ chosen in the Exchange step is not
$n-t$. Since generalized Hermite reduction does not introduce any
new zeros on the diagonal, and from the Exchange step and the Corner
step, we have that no diagonal element of $H(k)$ is zero. This
contradicts the hypothesis on $k$, and hence that our assumption
that $r < n - t$ was false. Thus  $r = n - t$ after the $(k - 1)$-th
iteration.

Next we show that there must be an SIR for $X$ appeared as a column
of the matrix $B$. We have $X^TH_{X} = \mathbf{0}$ from Lemma
\ref{lem:property-of-Hx} and hence that  ${\mathbf 0} =
X^TBB^{-1}H_{X} = X^TBB^{-1}H_{X}Q = X^TBH(k - 1)$, where $Q$ is an
appropriate orthogonal $(n - t)\times(n - t)$ matrix. Let $({\mathbf
z}_1, \cdots, {\mathbb Z}_t)^T = X^TB$, where ${\mathbf
z}_i=(z_{i,1},\cdots,z_{i,n})^T$. Then
\begin{equation}\label{eq:correctness}
\begin{split}&\,\,\,\,\,\,\left[\begin{array}{lll}0&\cdots&0\\
\vdots&\ddots &\vdots\\
0&\cdots&0\end{array}
 \right] = X^TBH(k - 1) = \left[
 \begin{array}{c}{\mathbb Z}_1^T\\\vdots\\{\mathbb Z}_t^T\end{array}
 \right]H(k - 1)\\&= \left[\begin{array}{lc}\cdots,&\sum_{k=n-t}^{n}z_{1,k}h_{k,n-t}(k-1)\\\cdots,&\cdots
\\\cdots,&\sum_{k=n-t}^{n}z_{t,k}h_{k,n-t}(k-1)\end{array}
 \right ]= \left [\begin{array}{lc}\cdots,&z_{1,n-t}h_{n-t, n-t}(k - 1)\\\cdots,&\cdots\\
 \cdots,&z_{t,n-t}h_{n-t, n-t}(k - 1)
 \end{array}
 \right ].\end{split}
 \end{equation}

We know $h_{n-t+1, n-t}(k - 1) =0$ and $h_{n-t, n-t}(k - 1) \neq 0$
from $h_{n-t, n-t}(k) = 0$. From Definition
\ref{def:Generalized-Hermite-reduction}
 and $h_{n-t+1, n-t}(k - 1)
=0$ we have $h_{n-t+2, n-t}(k - 1) = \cdots = h_{n, n-t}(k - 1) = 0$
which implies the last equality in (\ref{eq:correctness}). Since
$h_{n-t, n-t}(k - 1) \neq 0$, it follows that $z_{1, n-t} = \cdots =
z_{t, n-t} = 0$. Thus the $(n - t)$-th column of $B$ is an SIR for
$X$.
\end{proof}

From the analysis above and Lemma \ref{lem:Hn,n-2=0}, the
correctness of Algorithm 1 is proved. From the iteration of
Algorithm 1, $\|H(k)\|_F$ is decreasing with respect to $k$. Thus if
there exist SIRs for $X$, Algorithm 1 can always find one.

\begin{defn}[$\Pi$ function]\label{def:Pi-function}
Let $\lambda(X)$ be the least norm of SIRs for $X$. For the $k$-th
iteration in Algorithm 1, define
\[
\Pi(k) = \prod_{1\leq j\leq n-t} \min\left\{\gamma^{n-t}\lambda(X),
\frac{1}{\left|h_{j,j}(k)\right|}\right\}^{n-j}.
\]
\end{defn}

\begin{lem}\label{lem:Pi-function}
For $k>1$ we have
\begin{enumerate}
\item

$\left(\gamma^{n-t}\lambda(X)\right)^{\left[\left(
                                                          \begin{array}{c}
                                                            n \\
                                                            2 \\
                                                          \end{array}
                                                        \right)-\left(
                                                          \begin{array}{c}
                                                            t \\
                                                            2 \\
                                                          \end{array}
                                                        \right)
\right]}\geq\Pi(k)\geq 1$.
\item

$\Pi(k)\geq \sqrt{\frac{4\gamma^2}{\gamma^2+4}}\, \Pi(k-1)$.
\end{enumerate}
\end{lem}

The routine of analyzing the number of iterations in \cite{FBA1999}
can be carried over here with redefining the $\Pi$ function as
above, so we state Lemma \ref{lem:Pi-function} directly without
proof. From this lemma, it follows that the $\Pi$ function is
increasing with respect to $k$ and has an upper bound for a fixed
$\gamma\in(2/\sqrt{3},+\infty)$. From Definition
\ref{def:Pi-function} we can infer $\Pi(0)\geq 1$. And from Lemma
\ref{lem:Pi-function} we know that
\[\left(\gamma^{n-t}\lambda(X)\right)^{\left[\left(
                                                          \begin{array}{c}
                                                            n \\
                                                            2 \\
                                                          \end{array}
                                                        \right)-\left(
                                                          \begin{array}{c}
                                                            t \\
                                                            2 \\
                                                          \end{array}
                                                        \right)
\right]}\geq\Pi(k)\geq \left(\sqrt{\frac{4\gamma^2}{\gamma^2+4}}\,
\right)^k.\] Solving $k$ from this inequality gives the part 1 of
Theorem \ref{thm:Correctness}, as was to be shown.

\section{Empirical Study and Further Discussion}\label{sec:Empirical study}

\subsection{Implementation}

All discussions above are based on exact arithmetic operation, i.e.
uses the Blum-Shub-Smale model \cite{BSS1989, BCS1998} of
computation. The reason is that Algorithm 1 involves real numbers.
Thus Algorithm 1 can only be implemented using floating point
arithmetic on computer. Both Algorithm 1 and the HJLS simultaneous
relations algorithm when $t=2$, i.e. detecting an SIR for two real
vectors, were implemented in \emph{Maple} 13 under multiprecision
floating point arithmetic (one level scheme). Like PSLQ (see
\cite[Section 5]{BB2001}), it is possible to perform most iterations
using hardware floating point arithmetic, with only occasional
depending on multiprecision arithmetic.  So the two level
implementation of Algorithm 1 is also developed by the first author.
It partially uses software
floating arithmetic and mainly uses hardware floating arithmetic  (Cf. \cite[Section 5]{BB2001} for details). \footnote{ The package is available from {
\url{http://cid-5dbb16a211c63a9b.skydrive.live.com/self.aspx/.Public/sird.rar}}.}

It is well known that the Gram-Schmidt orthogonalization algorithm
is numerically unstable \cite{GL1989}, so in our implementations we
construct the hyperplane matrix by using QR decomposition, instead
of Gram-Schmidt orthogonalization. In contrast to HJLS, the
iteration step in Algorithm 1 is not based on the Gram-Schmidt
orthogonalization, but on LQ decomposition (this is equivalent to QR
decomposition). Householder transformations are used in our
implementations to compute these decompositions. Thus our
implementations of Algorithm 1 is numerically stable.

\subsection{Experimental Result} In theory, the cost of Algorithm 1 (in Corollary
\ref{cor:time-complexity-of-simultaneous-relation}) matches the best
complexity upper bound known for this problem (Cf. \cite[section
5]{HJL1989}), whereas in practice  Algorithm 1 usually needs fewer
iterations. For ${\mathbf x}_1 = (11,27,31)^T$ and ${\mathbf x}_2 =
(1,2,3)^T$, HJLS outputs $(19, -2, -5)^T$ after 5 iterations while
Algorithm 1 outputs the same SIR after only 3 iterations.

\begin{table}[h]
\centering
\begin{minipage}[c]{0.5\textwidth}
\centering
\begin{tabular}{|c|c|c|c|c|c|c|r|}
  \hline
No.&  $n$  & $itr_{HJLS}$ & $itr_{SIRD}$ & $t_{HJLS}$ & $t_{SIRD}$
 \\\hline
  1&4&15&8&\ \ 0.063&0. \\\hline
 2&4&13&6&\ \ 0.062&0. \\\hline
 3&4&21&11&\ \ 0.094&\ \ 0.015\\\hline
 4&5&25&12&\ \ 0.109&\ \ 0.016\\\hline
 5&5&27&7&\ \ 0.141&0.  \\\hline
 6& 5 &   21 & 10 & \ \ 0.094 &0. \\\hline
7& 30 &  51 & 7 & \ \ 0.922 & \ \ 0.125 \\\hline
  8&54 & 34 & 9 & \ \ 2.203 & \ \ 0.453 \\\hline
   9&79  & 34 & 5 & \ 4.860 & \ \ 0.625 \\\hline
  10&97 &  37 & 5 & \ 7.438 & \ \ 1.047 \\
  \hline
\end{tabular}
\end{minipage}%
\begin{minipage}[c]{0.5\textwidth}
\centering
\begin{tabular}{|c|c|c|c|c|c|c|c|}
  \hline
No.&  $n$  & $itr_{HJLS}$ & $itr_{SIRD}$ & $t_{HJLS}$ & $t_{SIRD}$
 \\\hline

  11&128 &  45 & 5 & \ 13.765 &  1.687 \\\hline
  12&149 &   29 & 2 & \ 19.016 & 1.610 \\\hline
   13&173 & 26 & 3 & \ 26.812 &  2.421 \\\hline
  14&192 &  29 & 5 &\  34.218 &  3.563 \\\hline
  15&278& 28 & 5 & \ 85.797 & 8.860  \\\hline
   16&290 & 35 & 4 & \ 95.656 & 8.328  \\\hline
   17&293 & 23 & 4 & \ 98.062 & 8.750  \\\hline
   18&305 & 22 & 3 &  109.187& 8.063  \\\hline
   19&316 & 19 & 3 & 120.187 & 8.766  \\\hline
  20&325 & 18 & 2 & 129.031 & 6.953  \\
  \hline
\end{tabular}
\end{minipage}
\caption{Comparison of performance results for HJLS and Algorithm 1
}\label{tab:performance-of-SIRD}
\end{table}

The purpose of the trials in Table \ref{tab:performance-of-SIRD} is
to compare the performances of HJLS and Algorithm 1 when $t=2$.  All
of the tests were run on AMD Athlon$^{\tiny\mbox{TM}}$ 7750
processor (2.70 GHz) with 2GB main memory.

In Table \ref{tab:performance-of-SIRD}, $n$ gives the dimension of
the relation vector, $itr_{HJLS}$ and $itr_{SIRD}$ are the numbers
of iterations of HJLS and Algorithm 1 respectively, and the columns
headed $t_{HJLS}$ and $t_{SIRD}$ give the CPU run time respectively
of the two algorithms in seconds. The 20 trials in Table
\ref{tab:performance-of-SIRD} were constructed by \emph{Maple}'s
pseudo random number generator. The first $6$ trials are for low
dimension, and   others for higher dimension.

The results show that Algorithm 1 appears to be more effective than
HJLS. In all  $20$ trials, the number of iterations of Algorithm 1
is less than that of HJLS. It is still true that Algorithm 1 usually
needs fewer iterations than  HJLS for more tests. This leads that
the running time of Algorithm 1 is much less than HJLS. With the
dimension $n$ increasing, the difference between the efficiencies of
Algorithm 1 and HJLS is increasingly notable. On average, the
running time of Algorithm 1 is less than $1/10$ (based on the data
in Table \ref{tab:performance-of-SIRD}) of the running time of HJLS.

To some extent, the number of iterations and the cost of Algorithm 1
are related to the parameter $\gamma$. In practice we have found
that for many examples, larger values of $\gamma$ are more effective
in finding SIRs for $X$. For ${\mathbf x}_1=(86,  6, 8, 673)^T$ and
${\mathbf x}_2=(83, 5, 87, 91)^T$, if we choose $\gamma=2$ then
Algorithm 1 outputs $(-32, -747, 63, 10)^T$ after 10 iterations,
however, if we choose $\gamma=93$, Algorithm 1 outputs $(-35, -2624,
157, 26)^T$ after only 7 iterations. It is worth mentioning that, in
the example above, both of the two different output vectors are SIRs
for $({\mathbf x}_1,{\mathbf x}_2)$ and they are linearly
independent and hence that all SIRs for $({\mathbf x}_1,{\mathbf
x}_2)$ can be obtained from them. Going on  choosing a $\gamma$
larger than $93$ in this example, after many tests, the authors find
that the number of iterations is always 7, and it will not decrease
any more. Based on this observation, all results in Table
\ref{tab:performance-of-SIRD} are obtained under the condition that
$\gamma=1000$.

In general, a larger $\gamma$ requires a higher precision in
authors' tests. Usually, a high precision leads a large height of
the output SIR and a large cost of memory. If one sets $\gamma =
1.15470053838$ $\left(>2/\sqrt{3}\right)$, for about $60\%$ of our
whole tests, the height of the vector returned by Algorithm 1 is
less than that of HJLS simultaneous relations algorithm, however the
number of iterations turns large.

This means that we should try to find the balance between the number
of iterations and the precision because both of them are relevant to
the running time of Algorithm 1. So in practice, what are the best
choices for $\gamma$ needs further exploration.
\subsection{An Application}
 We end this paper with an application of Algorithm 1 to
find the minimal polynomial of a complex algebraic number from its
approximation.

\begin{exmp}
Let $\alpha = 2 + \sqrt{3}I$. We know that the minimal polynomial of
$\alpha$ in ${\mathbb Z}[x]$ is $7-4x+x^2$. Let $\bar{\alpha} =
2.000 + 1.732 I$ be an approximation to $\alpha$ with four
significant digits. Let ${\mathbf v}_1 = (1., 2., 1.)^T$ and
${\mathbf v}_2 = (0., 1.732, 6.928)^T$ be the real part and the
imaginary part of $(1, \,\bar{\alpha},\,\bar{\alpha}^2)^T$
respectively. Feeding Algorithm 1 with ${\mathbf v}_1$, ${\mathbf
v}_2$ as its input vectors gives an SIR for ${\mathbf v}_1$,
${\mathbf v}_2$ after $2$ iterations. The corresponding matrices $B$
are
\[
 \left[ \begin {array}{rrr} 2&1&\,\,0\\\noalign{\medskip}-1&-1&0
\\\noalign{\medskip}0&0&1\end {array} \right],
\left[ \begin {array}{rrr}
7&\,\,0&2\\\noalign{\medskip}-4&0&-1\\\noalign{\medskip}1&1&0\end
{array} \right].
\]

It is obvious that the first column of the latter one is an SIR for
${\mathbf v}_1$ and ${\mathbf v}_2$ and corresponds to the
coefficients of the minimal polynomial of $\alpha$. However, if one
takes only $3$ significant digits for the same data, after $3$
iterations Algorithm 1 outputs $(1213,-693,173)^T$, which is an
exact SIR for $(1., 2., 1.)^T$ and $(0., 1.73, 6.93)^T$, but does
not correspond to the coefficients of the minimal polynomial of
$\alpha$. For this reason, how to appropriately control the error
also is an interesting problem.
\end{exmp}

Generally, for computing the minimal polynomial of an algebraic
number $\alpha$ with degree $n$, we detect an integer relation for
$(1,\alpha,\cdots,\alpha^n)^T$. If $\alpha\in\mathbb{C}$, we detect
an SIR for $(1,\mbox{Re}(\alpha),\cdots,\mbox{Re}(\alpha^n))^T$ and
$(0,\mbox{Im}(\alpha),\cdots,\mbox{Im}(\alpha^n))^T$ by Algorithm 1
under a proper decimal precision. The output vector corresponds to a
polynomial of degree $n$, whose primitive part must be the minimal
polynomial of $\alpha$.

As mentioned early, Algorithm 1 has been implemented in two schemes
(one-level, two-level). Using our two level implementation of
Algorithm 1, the authors obtain the following polynomial of degree
84 from an approximation to $\alpha=3^{1/6}-2^{1/7}I$ with
\texttt{Digits:=1300} in \emph{Maple} 13. It is easy to check that
this polynomial is the exact minimal polynomial of $\alpha$.
\[\begin{array}{ll}
&5067001+783962907\,{x}^{36}+21027764272536\,{x}^{40}-7504504\,{x}^{42}
\\+&83639618394696\,{x}^{34}-36683081862336\,{x}^{38}+770305668258672\,{x
}^{32}\\
+&142394998636968\,{x}^{28}+1254656434122\,{x}^{30}+1370000831472
\,{x}^{10}\\
+&34207465357611\,{x}^{12}-284692059376032\,{x}^{14}+
24758141678424\,{x}^{16}\\+&190959510258972\,{x}^{18}+2306173886216928\,{
x}^{20}+99120704967648\,{x}^{22}\\+&64111149001809\,{x}^{24}-
3029254676588448\,{x}^{26}+250312437648\,{x}^{44}\\
+&2189187\,{x}^{48}-
112615776\,{x}^{50}-486486\,{x}^{54}+81081\,{x}^{60}-378550368\,{x}^{2
}\\+&11935794528\,{x}^{4}+190431110646\,{x}^{6}+3293025660288\,{x}^{8}+{x
}^{84}\\+&88074554904\,{x}^{52}+240\,{x}^{56}+1041237288\,{x}^{58}+
1952496\,{x}^{64}-9828\,{x}^{66}\\+&24\,{x}^{70}+819\,{x}^{72}-42\,{x}^{
78}+2212809521832\,{x}^{46}
\end{array}
\]

Some performance results are reported in Table
\ref{tab:TLSIRD-VS-SIRD}. In this table, $r$ and $s$ ($s$ is an odd
integer number) define the constant $\alpha=3^{1/r}-2^{1/s}I$, which
is an algebraic number of degree $2rs$, and $n=2rs+1$. The column
headed ``\texttt{Digits}'' gives a sufficient precision in decimal
digits, while $itr$ and $t$ are the number of iterations and running
times required for the correct output  respectively, where the
suffix SIRD is for one-level and TLSIRD for two-level. Every output
vector in Table \ref{tab:TLSIRD-VS-SIRD} corresponds to the
coefficients of the exact minimal polynomial of $\alpha$'s.
\begin{table}\centering
\begin{tabular}{|c|c|c|c|c|r|c|r|}
  \hline
$r$& $s$&  Dim. $n$  & \texttt{Digits} & $itr_{SIRD}$ &$t_{SIRD}\ \
$ & $itr_{TLSIRD}$ & $t_{TLSIRD}$
 \\\hline
4&3&25&100& 5685&75.937&5611&8.125\\\hline
3&5&31&150&11792&356.890&11792&28.157\\\hline
6&3&37&300&18927&556.031&18993&73.109\\\hline
4&5&41&350&26600&942.516&26192&134.360\\\hline
5&5&51&400&50084&2432.672&49738&440.110\\\hline
6&5&61&550&84677&6422.079&81758&1267.985\\\hline 4&9&73&1000& & &
159326&4889.922\\ 6&7&85&1300& & & 234422&10658.735\\\hline
\end{tabular}
\caption{Running times for the two implementations of Algorithm
1}\label{tab:TLSIRD-VS-SIRD}
\end{table}

From Table \ref{tab:TLSIRD-VS-SIRD}, it can be seen that the
two-level program is much faster than the one-level program since
the two-level program  only involves multiprecision operation
partially. The reason is that
 not all operations of Algorithm 1 have
to use high precision. As a matter of fact, the two-level program
performs most of the iterations using IEEE hardware arithmetic. Thus
the running times can be dramatically reduced.

Using the case of $t=1$ (it is PSLQ in fact) and $t=2$ of Algorithm
1 for real and complex algebraic numbers, respectively, we get a new
complete method (The method in \cite{QFC2009} is only for real
algebraic numbers.) to recover the minimal polynomial of an
arbitrary algebraic number from its approximation and degree.  It
should be noted that since this method depends on integer relation
detection that is based on a generalization of Euclidean algorithm
\cite{FF1979}, it is different from LLL-based algorithms, such as
\cite{KLL1988,HN2010}.

In practice, the  presented method is efficient. For
$\alpha:=\sqrt[2]{21}+\sqrt[3]{43}I$, our  procedure
\texttt{MiniPoly} takes $1.062$ seconds for outputting the exact
minimal polynomial of $\alpha$, whereas the \emph{Maple} built-in
function \texttt{MinimalPolynomial} in \texttt{PolynomialTools}
package that is LLL-based takes $6.032$ seconds under the same
decimal precision \texttt{Digits:=500}.

\section{Conclusion}
Using a method to construct a \emph{hyperplane matrix} and \emph{the
generalized Hermite reduction}, a new SIRs detecting algorithm,
Algorithm 1, is presented in this paper. It runs faster than the
HJLS simultaneous relation algorithm through the authors'
\emph{Maple} package. Applying the algorithm, we obtain a complete
method to find the minimal polynomial of an approximately algebraic
number, which is even faster than the corresponding \emph{Maple}
built-in function.

\ \\ {\bf Acknowledgements.} This research was partially supported
by the Knowledge Innovation Program of Cinese Academy of Sciences
(KJCX2-YW-S02), the National Natural Science Foundation of China
(10771205) and the National Basic Research Program of China
(2011CB302400).
%
%

\begin{thebibliography}{10}

\bibitem{BB2001}
D.~H. Bailey and D.~J. Broadhurst.
\newblock Parallel integer relation detection: Techniques and applications.
\newblock {\em Math. Comput.}, 70(236): 1719 --1736, 2000.

\bibitem{BB2009}
David~H. Bailey and J.M. Borwein.
\newblock PSLQ: An algorithm to discover integer relations, 2009.
\newblock available from \url{http://escholarship.org/uc/item/95p4255b}.

\bibitem{BBC2007}
David~H. Bailey, Jonathan~M. Borwein, Neil~J. Calkin, Roland
Girgensohn,
  D.~Russell Luke, and Victor~H. Moll.
\newblock {\em Experimental Mathematics in Action}.
\newblock AK Peters, 2007.

\bibitem{BCS1998}
L. Blum, F. Cucker, M. Shub, and S. Smale.
\newblock {\em Complexity and real computation}.
\newblock Springer-Verlag, 1998.

\bibitem{BSS1989}
L. Blum, M. Shub, and S. Smale.
\newblock {On a theory of computation
and complexity over the real numbers: NP-completeness, recursive
functions and universal machines}.
\newblock {\em Bulletin of the American Mathematical
Society}, 21(1): 1--46, 1989.

\bibitem{DS2000}
J.~Dongarra and F.~Sullivan.
\newblock Guest editors' introduction: the top 10 algorithms.
\newblock {\em Comput. Sci. Eng.}, 2(1): 22--23, 2000.

\bibitem{FB1992}
H.~R.~P. Ferguson and D.~H. Bailey.
\newblock Polynomial time, numerically stable integer relation algorithm.
\newblock Technical Report RNR-91-032, NAS Applied Research Branch, NASA Ames
  Research Center, Mar. 1992.

\bibitem{FBA1999}
H.~R.~P. Ferguson, D.~H. Bailey, and S.~Arno.
\newblock Analysis of PSLQ, an integer relation finding algorithm.
\newblock {\em Math. Comput.}, 68(225): 351--369, 1999.

\bibitem{FF1979}
H.~R.~P. Ferguson, and R. W. Forcade.
\newblock Generalization of the Euclidean algorithm for real nubmers
to all dimensions higher than two.
\newblock{\em Bulletin of the American Mathematical
Society}, 1(6): 912-914, 1979.

\bibitem{GL1989}
G. H. Golub and C. F. Van Loan.
\newblock {\em Matrix Computations}.
\newblock {The Johns Hopkins University Press}, 3rd edition, 1996.

\bibitem{HJL1989}
J.~Hastad, B.~Just, J.~C. Lagarias, and C.~P. Schnorr.
\newblock Polynomial time algorithms for finding integer relations among real
  numbers.
\newblock {\em SIAM Journal on Computing}, 18(5): 859--881, 1989.

\bibitem{HN2010}
Mark van Hoeij and Andrew Novocin.
\newblock Gradual sub-lattice reduction and a new complexity for factoring
polynomials. \newblock{ In: \em LATIN '10}, LNCS 6034: 539--553,
2010.

\bibitem{KLL1988}
R. Kannan, A. K. Lenstra~and~L. Lov{\'a}sz, L. \newblock Polynomial
Factorization and Nonrandomness of Bits of Algebraic and Some
Transcendental Numbers. \newblock{\em Math.Comput.}, 50(182):
235--250, 1988.

\bibitem{LLL1982}
A.~K. Lenstra, H.~W. Lenstra, and L.~Lov{\'a}sz.
\newblock Factoring polynomials with rational coefficients.
\newblock {\em Math. Ann.}, 261(4): 515--534, 1982.

\bibitem{QFC2009}
Xiaolin Qin, Yong Feng, Jingwei Chen and Jingzhong Zhang.
\newblock Finding Exact Minimal Polynomial by Approximations.
\newblock In: {\em SNC'09}, 125--131, 2009.

\bibitem{RS1997}
Carsten R{\"o}ssner and C.~P. Schnorr.
\newblock Diophantine approximation of a plane, 1997.
\newblock available from \url{http://citeseer.ist.psu.edu/193822.html}.

\end{thebibliography}

\end{document}